\documentclass{llncs}
\usepackage[ansinew]{inputenc}
\usepackage{graphicx}
\usepackage[english]{babel}
\usepackage{mathtools,amssymb}
\usepackage[left,pagewise] {lineno}
\usepackage {xspace}
\usepackage {url}
\usepackage {plaatjes}
\usepackage {times}
\usepackage {xcolor}
\usepackage[textsize=footnotesize]{todonotes}
\usepackage {microtype}
\usepackage[vlined,ruled]{algorithm2e}
\usepackage {wrapfig}

\let\doendproof\endproof
\renewcommand\endproof{~\hfill$\square$\doendproof}

\definecolor {sepia} {hsb} {5,0.75,0.65}
\renewcommand{\paragraph}[1]{\medskip\noindent\textbf{#1.}}

\title{Strict Confluent Drawing}

\author
{
  David Eppstein\inst{1}
  \and Danny Holten\inst{2}
  \and Maarten L\"offler\inst{3},
  \\ Martin N\"ollenburg\inst{4}
  \and Bettina Speckmann\inst{5}
  \and Kevin Verbeek\inst{6}
  }

\institute{Computer Science Department, University of California, Irvine, USA, \email{eppstein@uci.edu}
\and Synerscope BV, Eindhoven, the Netherlands, \email{danny.holten@synerscope.com}
\and Department of Computing and Information Sciences, Utrecht University, the Netherlands, \email{m.loffler@uu.nl}
\and Institute of Theoretical Informatics, Karlsruhe Institute of Technology, Germany, \email{noellenburg@kit.edu}
\and Department of Mathematics and Computer Science, Technical University Eindhoven, the Netherlands, \email{speckman@win.tue.nl}
\and Department of Computer Science, University of California, Santa Barbara, USA, \email{kverbeek@cs.ucsb.edu}
}

\begin{document}

\maketitle

\begin{abstract}
We define \emph{strict confluent drawing}, a form of confluent drawing in which the existence of an edge is indicated by the presence of a smooth path through a system of arcs and junctions (without crossings), and in which such a path, if it exists, must be unique. We prove that it is NP-complete to determine whether a given graph has a strict confluent drawing but polynomial to determine whether it has an \emph{outerplanar} strict confluent drawing with a fixed vertex ordering (a drawing within a disk, with the vertices placed in a given order on the boundary).
\end{abstract}

\section {Introduction}

Confluent drawing is a style of graph drawing in which edges are not drawn explicitly; instead vertex adjacency is indicated by the existence of a smooth path through a system of arcs and junctions that resemble train tracks. These types of drawings allow even very dense graphs, such as complete graphs and complete bipartite graphs, to be drawn in a planar way~\cite{degm-cd-05}.
Since its introduction, there has been much subsequent work on confluent drawing~\cite{EppGooMen-Alg-07,egm-dcd-06,EppSim-GD-11,hmr-becgcd-06,hss-ttcd-04,qa-cdard-10}, but the complexity of confluent drawing has remained unclear: how difficult is it to determine whether a given graph has a confluent drawing?
Confluent drawings have a certain visual similarity to a graph drawing technique called \emph {edge bundling}~\cite {Cui2008,Dwyer2007,Holten2006,Holten2009,Hurter2012}, in which ``similar'' edges are routed together in ``bundles'',
but we note that these drawings should be interpreted differently. In particular, sets of edges bundled together form visual junctions, however, interpreting them as confluent junctions can create false adjacencies.

Formally, a confluent drawing may be defined as a collection of \emph {vertices}, \emph {junctions} and \emph {arcs} in the plane, such that all arcs are smooth and start and end at either a junction or a vertex, such that arcs intersect only at their endpoints, and such that  all arcs that meet at a junction share the same tangent line there. A confluent drawing $D$ represents a graph $G$ defined as follows: the vertices of $G$ are the vertices of $D$, and there is an edge between two vertices $u$ and $v$ if and only if there exists a smooth path in $D$ from $u$ to $v$ that does not pass any other vertex. (In some variants of confluent drawing an additional restriction is made that the smooth path may not intersect itself~\cite{hss-ttcd-04}; however, this constraint is not relevant for our work.)

\begin{wrapfigure}[9]{r}{0.29\textwidth}
    \vspace{-\baselineskip}
		\vspace{-5ex}
    \centering
    \subfigure[\label{fig:mult-adj}]{
		\includegraphics[scale=1]{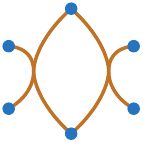}}
		\hfil
		\subfigure[\label{fig:self-loop}]{
		\includegraphics[scale=1]{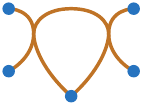}}
	\caption{%
		  (a) A drawing with a duplicate path.
		  (b)~A drawing with a self-loop.}
	\label{fig:mult-adj+self-loop}
\end{wrapfigure}
\paragraph {Contribution}
In this paper we introduce a subclass of confluent drawings, which we call \emph {strict} confluent drawings. Strict confluent drawings are confluent drawings with the additional restrictions that between any pair of vertices there can be \emph {at most one} smooth path, and there cannot be any paths from a vertex to itself.
Figure~\ref {fig:mult-adj+self-loop} illustrates the forbidden configurations. To avoid irrelevant components in the drawing, we also require all arcs of the drawing to be part of at least one smooth path representing an edge.
We believe that these restrictions may make strict drawings easier to read, by reducing the ambiguity caused by the existence of multiple paths between vertices. In addition, as we show, the assumption of strictness allows us to completely characterize their complexity, the first such characterization for any form of confluence on arbitrary undirected graphs.

We prove the following:
\begin{itemize}
\item It is NP-complete to determine whether a given graph has a strict confluent drawing.
\item For a given graph, with a given cyclic ordering of its vertices, there is a polynomial time algorithm to find an \emph{outerplanar} strict confluent drawing, if it exists: this is a drawing in a disk, with the vertices in the given order on the boundary of the disk
\item When a graph has an outerplanar strict confluent drawing, an algorithm based on circle packing can construct a layout of the drawing in which every arc is drawn using at most two circular arcs.
\end{itemize}

\begin{figure}[b]
	\centering
		\subfigure[\label{fig:gd2011contest}]{%
		\includegraphics[scale=0.33]{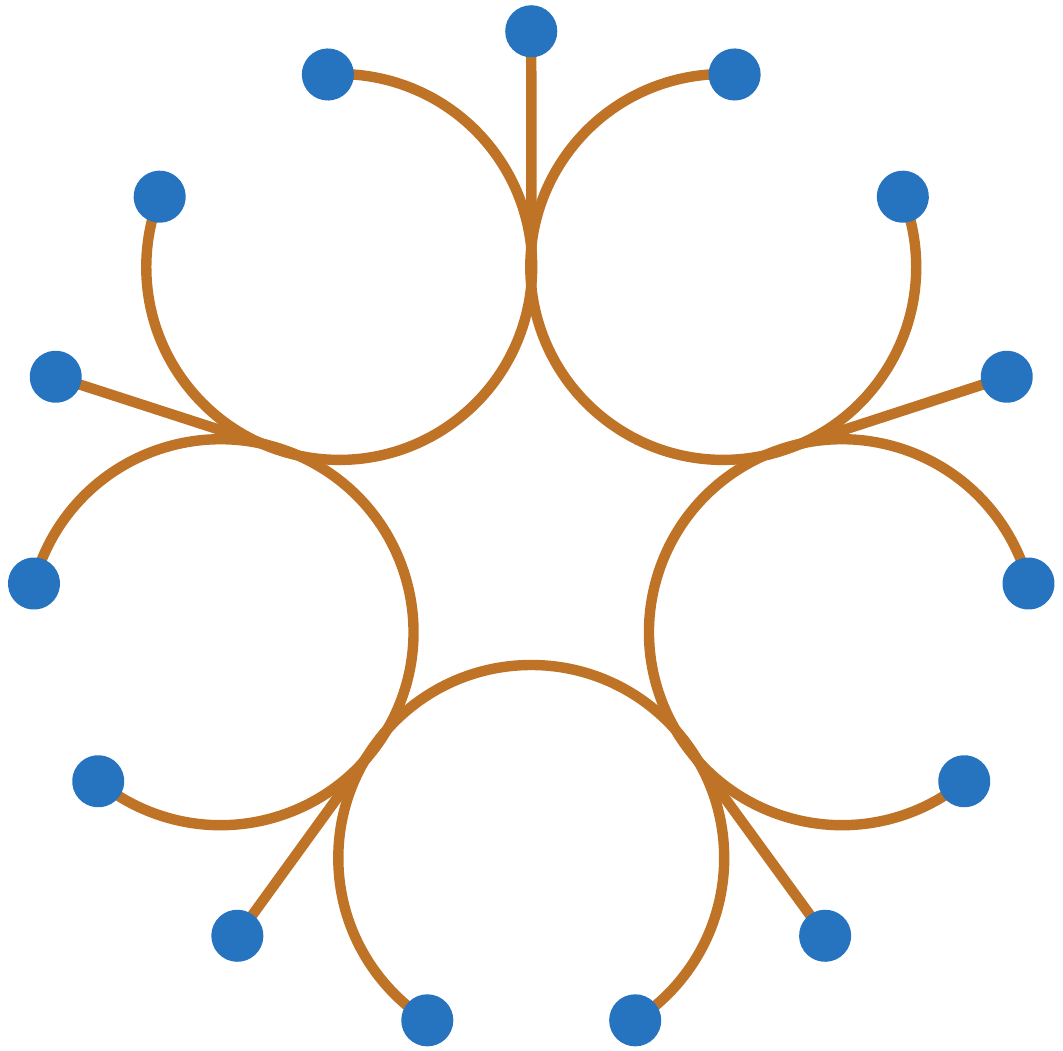}}
		\hfil
		\subfigure[\label{fig:5wheel-new}]{
		\includegraphics[page=2]{figures/5wheel-new}}
	\caption{(a) Outerplanar strict confluent drawing of the GD2011 contest graph. (b) A graph with no outerplanar strict confluent drawing.}
	\label{fig:sea-creature-5wheel-new}
\end{figure}

See Fig.~\ref{fig:gd2011contest} for an example of an outerplanar strict confluent drawing.  Previous work on  \emph{tree-confluent}~\cite{hss-ttcd-04} and \emph{delta-confluent drawings}~\cite{egm-dcd-06} characterized special cases of outerplanar strict confluent drawings as being the chordal bipartite graphs and distance-hereditary graphs respectively, so these graphs as well as the outerplanar graphs are all outerplanar strict confluent. The six-vertex wheel graph in Fig.~\ref{fig:5wheel-new} provides an example of a graph that does not have an outerplanar strict confluent drawing. (The central vertex $u$ needs to be placed between two of the outer vertices, say, $a$ and $b$. The smooth path from $u$ to the opposite vertex $d$ separates $a$ and $b$, so there must be a junction shared by the $u$--$d$ and $a$--$b$ paths, creating a wrong adjacency with $d$.)

\section {Preliminaries}

Let $G = (V,E)$ be a graph.
We call an edge $e$ in a drawing $D$ \emph {direct} if it consists only of a single arc (that does not pass through junctions).
We call the angle between two consecutive arcs at a junction or vertex \emph{sharp} if the two arcs do not form a smooth path; each junction has exactly two angles that are not sharp, and every angle at a vertex is sharp (so the number of sharp angles equals the degree of the vertex).

\begin {lemma} \label {lem:deg2}
  Let $G$ be a graph, and let $E' \subseteq E$ be the edges of $E$ that are incident to at least one vertex of degree $2$.
  If $G$ has a strict confluent drawing $D$, then it also has a strict confluent drawing $D'$ in which all edges in $E'$ are direct.
\end {lemma}

\begin{proof}
  Let $v$ be a degree-2 vertex in $G$ with two incident edges $e$ and $f$. We consider the representation of $e$ and $f$ in $D$ and modify $D$ so that $e$ and $f$ are single arcs. There are two cases. If $e$ and $f$ leave $v$ on two disjoint paths, then these paths have only merge junctions from $v$'s perspective. We can simply separate these junctions from $e$ and $f$ as shown in Fig.~\ref{sfg:sing-arc-disjoint}. If, on the other hand, $e$ and $f$ share the same path leaving $v$, then their paths split at some point. We need to reroute the merge junctions prior to the split and separate the merge junctions after the split as shown in Fig.~\ref{sfg:sing-arc-joint}. This is always possible since $v$ has no other incident edges. Because $D$ was strict and these changes do not affect strictness, $D'$ is still a strict confluent drawing and edges $e$ and $f$ are direct.
\end{proof}

\begin{figure}[htbp]
  \centering
  \subfigure[\label{sfg:sing-arc-disjoint}]{\includegraphics[page=1]{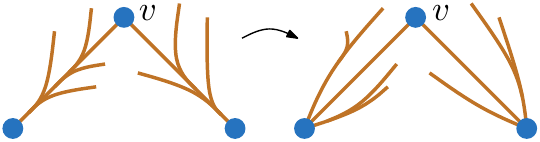}}
  \hfill  \subfigure[\label{sfg:sing-arc-joint}]{\includegraphics[page=2]{figures/single-arc}}
  \caption{The two cases of creating single arcs for edges incident to a degree-2 vertex.}
  \label{fig:single-arc}
\end{figure}

\begin {lemma} \label {lem:k22}
  Let $G$ be a graph. If $G$ has no $K_{2,2}$ as a subgraph, whose vertices have degrees $\ge 3$ in $G$, then $G$ has a strict confluent drawing if and only if $G$ is planar.
\end {lemma}

\begin{proof}
  Since every planar drawing is also a strict confluent drawing, that implication is obvious.
	So let $D$ be a strict confluent drawing for a graph $G$ without a $K_{2,2}$ subgraph, whose vertices have degrees $\ge 3$ in $G$. Since larger junctions, where more than three arcs meet, can easily be transformed into an equivalent sequence of binary junctions, we can assume that every junction in $D$ is binary, i.e., two arcs merge into one (or, from a different perspective, one arc splits into two).
	By Lemma~\ref{lem:deg2} we can further transform $D$ so that all edges incident to degree-$2$ vertices are direct.
	Now for any vertex $u$ in $D$ none of its outgoing paths to some neighbor $v$ can visit a merge junction before visiting a split junction as this would imply either a non-strict drawing or a $K_{2,2}$ subgraph with vertex degrees $\ge 3$. So the sequence of junctions on any $u$-$v$ path consists of a number of split junctions followed by a number of merge junctions. But any such path can be unbundled from its junctions to the left and right and turned into a direct edge without creating arc intersections  as illustrated in Fig.~\ref{fig:no-k22}. This shows that $D$ can be transformed into a standard planar drawing of~$G$.
\end{proof}

\begin{figure}[htbp]
	\centering
		\includegraphics[scale=1]{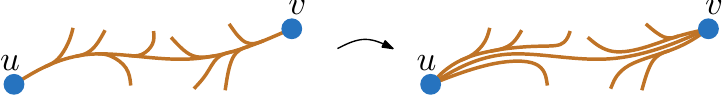}
	\caption{Any strict confluent drawing of a graph without a $K_{2,2}$ subgraph can be transformed into a standard planar drawing.}
	\label{fig:no-k22}
\end{figure}

The next lemma characterizes the combinatorial complexity of strict confluent drawings. Its proof is found in Appendix~\ref{app:complexity} and uses Euler's formula and double counting.

\begin{lemma}\label{lem:linear}
	The combinatorial complexity of any strict confluent drawing $D$ of a graph~$G$, i.e., the number of arcs, junctions, and faces in $D$, is linear in the number of vertices of~$G$. 
\end{lemma}

Lemma~\ref{lem:linear} is in contrast to previous methods for confluently drawing interval graphs~\cite{degm-cd-05} and for drawing confluent Hasse diagrams~\cite{EppSim-GD-11}, both of which may produce (non-strict) drawings with quadratically many features.

\section {Computational Complexity}

We will show by a reduction from planar 3-SAT~\cite{l-pftu-82} that it is NP-complete to decide whether a graph $G$ has a strict confluent drawing in which all edges incident to degree-$2$ vertices are direct. By Lemma~\ref {lem:deg2}, this is enough to show that it is also NP-complete to decide if $G$ has any strict confluent drawing.

Consider the subdivided grid graph (a grid with one extra vertex on each edge). In this graph, all edges are adjacent to a degree $2$ vertex. Since a grid graph more than one square wide has only one fixed planar embedding (up to choice of the outer face), the subdivided grid graph has only one confluent embedding in which all edges are direct. We will base our construction on a number of such grids.

\tweeplaatjes [scale=.78] {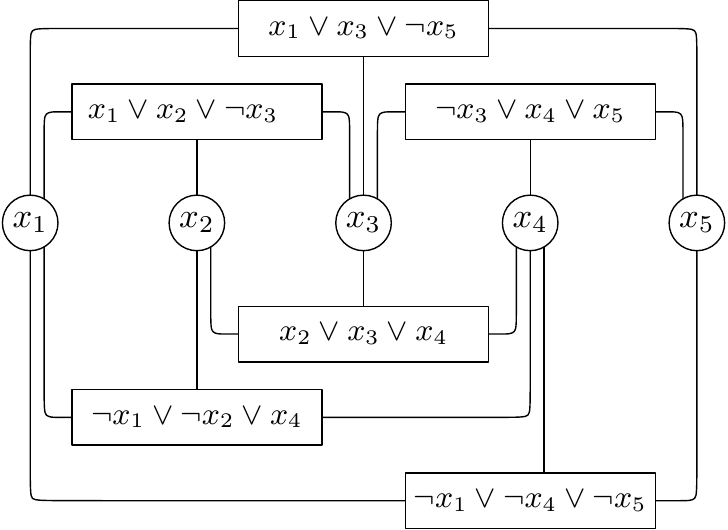} {hard-global-frame}
{ (a) A planar 3-SAT formula.
  (b) The corresponding global frame of the construction: one grid graph per variable, with some vertices identified at each clause. Green boundary edges correspond to positive literals, red edges to negated literals.
      For easier readability the grids in this figure are larger than strictly necessary.
}

Let $S$ be a planar 3-SAT formula. %
Globally speaking, we will create a grid graph for each variable of $S$, of size depending on the number of clauses that the variable appears in. The external edges of this grid graph are alternatingly colored green and red.
We connect the variable graphs by identifying certain vertices: for each of the three variables that appear in a clause, we select one subdivided edge (that is, three vertices connected by two edges) on the outer face, and identify the endpoints of these edges into a triangle of subdivided edges (that is, a $6$-cycle). We choose a green edge for a positive occurrence of the variable and a red edge for a negated occurrence. This will become clear below.
We call the resulting graph $F$ the \emph {frame} of the construction; all edges of $F$ are adjacent to a degree-$2$ vertex and $F$ has only one planar embedding (up to choice of the outer face).
Figure~\ref {fig:hard-global-3sat+hard-global-frame} shows an example.

\eenplaatje {hardness-K4} {$K_4$ and its two strict confluent drawings, without moving the vertices and keeping all arcs inside the convex hull of the vertices.}
\drieplaatjes [scale=.97] {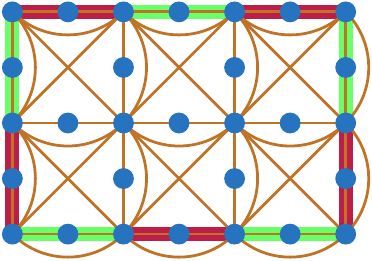} {hardness-var-true} {hardness-var-false}
{ (a) A variable gadget consists of a grid of $K_4$'s. Green (light) edges of the frame highlight normal literals, red (dark) edges negated ones.
  (b) One of the two possible strict  confluent drawings, corresponding to the value \emph {true}.
  (c) The other strict confluent drawing, corresponding to \emph {false}.
}

\begin{wrapfigure}[18]{r}{0.4\textwidth}
    \centering
    \includegraphics[scale=.8]{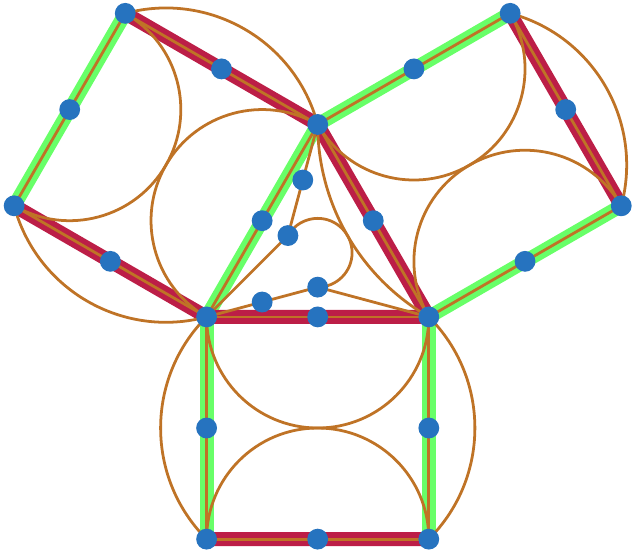}
	\caption{Three variables attached to a clause gadget. The top left variable occurs in the clause as a positive literal, the others as negative literals. The clause can be satisfied because the top right variable is set to \emph {false}.}
	\label{fig:hardness-clause-attached}
\end{wrapfigure}
The main idea of the construction is based on the fact that $K_4$, when drawn with all four vertices on the outer face, has exactly two strict confluent drawings: we need to create a junction that merges the diagonal edges with one pair of opposite edges, and we can choose the pair.
Figure~\ref {fig:hardness-K4} illustrates this.
We will add a copy of $K_4$ to every cell of the frame graph $F$. Recall that every cell, except for the triangular clause faces, is a subdivided square (that is, an $8$-cycle). We add $K_4$ on the four grid vertices (not the subdivision vertices). The edges that connect external grid vertices are called \emph{literal edges}.
Figure~\ref {fig:hardness-var-graph} shows this for a small grid.
Since neighboring grid cells share a (subdivided) edge, the $K_4$'s are not edge-independent. This implies that in a strict confluent drawing, we cannot ``use'' such a common edge in both cells. Therefore, we need to orient the $K_4$-junctions alternatingly, as illustrated in Figures~\ref {fig:hardness-var-true} and~\ref {fig:hardness-var-false}.
If the grid is sufficiently large (every cell is part of a larger at least size-$(2 \times 2)$ grid) these choices are completely propagated through the entire grid, so there are two structurally different possible embeddings,  which we use to represent the values \emph {true} and \emph {false} of the corresponding variable. For every green edge of the frame in the \emph{true} state and every red edge in the \emph{false} state there is one remaining literal edge in the outer face, which can still be drawn either inside or outside their grid cells. In the opposite states these literal edges are needed inside the grid cells to create the $K_4$ junctions.  The availability of at least one literal edge (corresponding to a \emph{true} literal) is important for satisfying the clause gadgets, which we describe next.

\vierplaatjes {hardness-clause} {hardness-clause-1} {hardness-clause-2} {hardness-clause-3}
{ (a) The input graph of the clause.
  (b, c, d) Three different strict confluent drawings.
}

Inside each triangular clause face, we add the graph depicted in Figure~\ref {fig:hardness-clause}. This graph has several strict confluent drawings; however, in every drawing at least one of the three outer edges needs to be drawn inside the subdivided triangle.

\begin {lemma}\label{lem:one-edge-in}
  There is no strict confluent drawing of the clause graph in which all three long edges are drawn outside. Moreover, there is a strict confluent drawing of the clause graph with two of these edges outside, for every pair.
\end {lemma}

\begin {proof}
  Recall that by Lemma~\ref {lem:deg2} the subdivided triangle must be embedded as a $6$-cycle of direct arcs.
To prove the first part of the lemma, assume that the triangle edges are all drawn outside this cycle.
  The remainder of the graph has no $4$-cycles without subdivision vertices (that is, no $K_{2,2}$ with higher-degree vertices), so by Lemma~\ref {lem:k22} it can only have a strict confluent drawing if it is planar. However, it is a subdivided $K_5$, which is not planar.  To prove the second part of the lemma, we refer to Figures~\ref {fig:hardness-clause-1},~\ref {fig:hardness-clause-2} and~\ref {fig:hardness-clause-3}.
\end {proof}

This describes the reduction from a planar 3-SAT instance to a graph consisting of variable and clause gadgets. Next we show that this graph has a strict confluent drawing if and only if the planar 3-SAT formula is satisfiable. For a given satisfying assignment we choose the corresponding embeddings of all variable gadgets. The assignment has at least one \emph{true} literal per clause, and correspondingly in each clause gadget one of the three literal edges can be drawn inside the clause triangle, allowing a strict confluent drawing by Lemma~\ref{lem:one-edge-in}. Conversely, in any strict confluent drawing, each clause must be drawn with at least one literal edge inside the clause triangle by Lemma~\ref{lem:one-edge-in}, so translating the state of each variable gadget into its truth value yields a satisfying assignment.

To show that testing strict confluence is in NP, recall that by Lemma~\ref{lem:linear} the combinatorial complexity of the drawing is linear in the number of vertices. Thus the existence of a drawing can be verified by guessing its combinatorial structure and verifying that it is planar and a drawing of the correct graph.

\begin {theorem}
  Deciding whether a graph has a strict confluent drawing is NP-complete.
\end {theorem}

\section {Outerplanar Strict Confluent Drawings}

For a graph $G$ with a fixed cyclic ordering of its vertices, we can test in polynomial time whether an outerplanar strict confluent drawing with this vertex ordering exists, and, if so, construct one. This algorithm uses the closely related notion of a canonical diagram of $G$, which is unique and exists if and only if an outerplanar strict confluent drawing exists. From the canonical diagram a confluent drawing can be constructed. We further show that the drawing can be constructed such that every arc consists of at most two circular arcs.

\subsection{Canonical Diagrams}

We define a \emph{canonical diagram} to be a collection of junctions and arcs connecting the vertices in the given order on the outer face (as in a confluent drawing), but with some of the faces of the diagram \emph{marked}, satisfying additional constraints enumerated below. Figure~\ref{fig:threeviews} shows a canonical diagram and an outerplanar strict confluent drawing of the same graph. In such a diagram, a \emph{trail} is a smooth curve from one vertex to another that follows the arcs (as in a confluent drawing) but is allowed to cross the interior of marked faces from one of its sharp corners to another. The constraints are:
\begin{itemize}
\item Every arc is part of at least one trail.
\item No two trails between the same two vertices can follow different sequences of arcs and faces.
\item Each marked face must have at least four angles, all of which are sharp.
\item Each arc must have either sharp angles or vertices at both of its ends.
\item For each junction $j$ with exactly two arcs in each direction, let $f$ and $f'$ be the two faces with sharp angles at $j$. Then it is not allowed for $f$ and $f'$ to both be either marked or to be a triangle (a face with three angles, all sharp).
\end{itemize}

\eenplaatje [width=\textwidth] {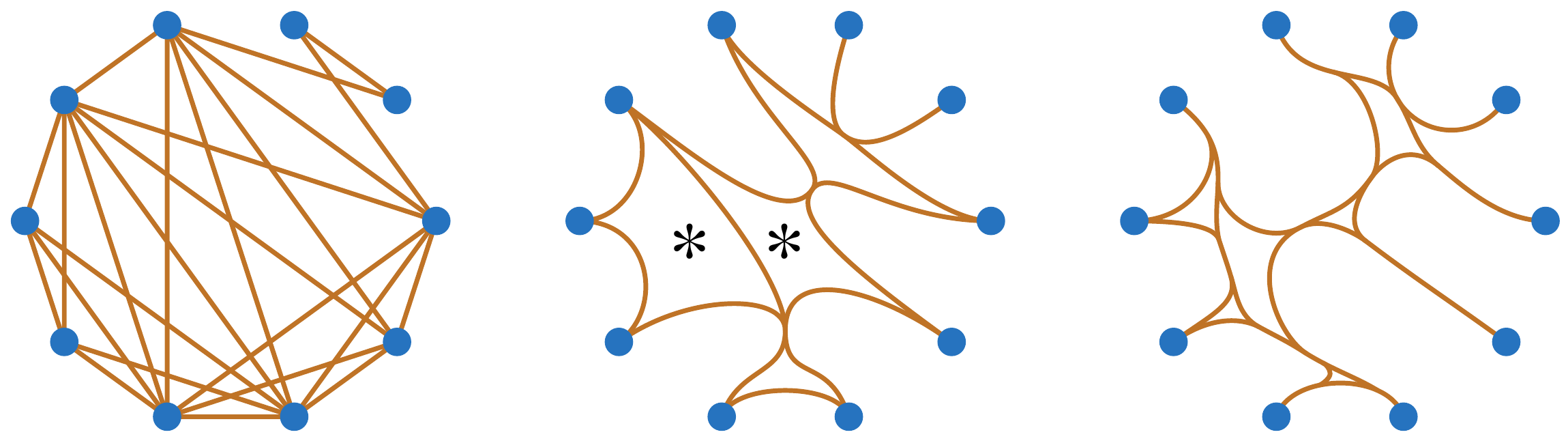} {Three views of the same graph as a node-link diagram (left), canonical diagram (center), and outerplanar strict confluent drawing (right).}

Let $j$ be a junction of a canonical diagram $D$. Then
define the \emph{funnel} of $j$ to be the 4-tuple of vertices $a,b,c,d$ where $a$
is the vertex reached by a path that leaves $j$ in one direction and
continues as far clockwise as possible, $b$ is the most counterclockwise
vertex reachable in the same direction from $j$, $c$ is the most clockwise
vertex reachable in the other direction, and $d$ is the most
counterclockwise vertex reachable in the other direction. Note that
none of the paths from $j$ to $a$, $b$, $c$, and $d$ can intersect each other without contradicting the uniqueness of trails. We call the circular intervals of vertices $[a,b]$ and $[c,d]$ (in the counterclockwise direction) the \emph{funnel intervals} of the respective funnel. We say a circular interval $[a,b]$ is \emph{separated} if either $a$ and $b$ are not adjacent in $G$, or there exists a junction in the canonical diagram with funnel intervals $[a,e]$ and $[f,b]$, where $e, f \in [a,b]$.

A canonical diagram represents a graph $G$ in which the edges in $G$ correspond to trails in the diagram. As we show in Appendix~\ref {app:canonical}, a graph $G$ has a canonical diagram if and only if it has an outerplanar strict confluent drawing, and if a canonical diagram exists then it is unique.

\subsection{Algorithm}

By using the properties of canonical diagrams (see Appendix~\ref{app:canonical}), we may obtain an algorithm that constructs a canonical diagram and strict confluent drawing of a given cyclically-ordered graph $G$, or reports that no drawing exists, in time and space $O(n^2)$. This bound is optimal in the worst case, as it matches the input size of a graph that may have quadratically many edges.

Steps~1--3 of the algorithm, detailed below, build some simple data structures that speed up the subsequent computations. Step~4 discovers all of the funnels in the input, from which it constructs a list of all of the junctions of the canonical diagram. Step~5 connects these junctions into a planar drawing, a subset of the canonical diagram. Step~6 builds a graph for each face of this drawing that will be used to complete it into the entire canonical diagram, and step~7 uses these graphs to find the remaining arcs of the diagram and to determine which faces of the diagram are marked. Step~8 checks that the diagram constructed by the previous steps correctly represents the input graph, and step~9 splits the marked faces, converting the diagram into a strict confluent drawing.

\begin{enumerate}
\item Number the vertices clockwise around the boundary cycle from $0$ to $n-1$.
\item\label{algorithm interval adjacency} Build a table, containing for each pair $i,j$, the number of ordered pairs $(i',j')$ with $i'\le i$, $j'\le j$, and vertices $i'$ and $j'$ adjacent in $G$. By performing a constant number of lookups in this table we may determine in constant time how many edges exist between any two disjoint intervals of the boundary cycle.
\item Build a table that lists, for each ordered pair $u,v$ of vertices, the neighbor $w$ of $u$ that is closest in clockwise order to $v$. That is, $w$ is adjacent to $u$, and the interval from $v$ clockwise to $w$ contains no other neighbors of $u$. The table entries for $u$ can be found in linear time by a single counterclockwise scan. Repeat the same construction in the opposite orientation.
\item\label{algorithm find junctions} For each separated interval $[a, b]$, let $c$ be the next neighbor of $a$ that is counterclockwise of $b$, and let $d$ be the next neighbor of $b$ that is clockwise of $a$. If (i) $c$ is a neighbor of $b$, (ii) $d$ is a neighbor of $a$, (iii) $a$ is the next neighbor of $c$ that is counterclockwise of $d$, and (iv) $b$ is the next neighbor of $d$ that is clockwise of $c$, then
(if a confluent diagram exists) $a,b,c,d$ must form the funnel of a junction, and all funnels have this form. We check all circular intervals in increasing order of their cardinalities. For each discovered funnel, we mark the intervals that are separated by the corresponding junction. This way we can check in $O(1)$ time whether a circular interval is separated.
If the number of funnels exceeds the linear bound of Lemma~\ref{lem:linear} on the number of junctions in a confluent drawing, abort the algorithm.
\item\label{algorithm funnel junctions} Create a junction for each of the funnels found in step~\ref{algorithm find junctions}. For each vertex $v$, make a set $J_v$ of the junctions whose funnel includes that vertex; if they are to be drawn as part of a canonical diagram, the junctions of $J_v$ need to be connected to $v$ by a confluent tree. For any two junctions in $J_v$, it is possible to determine in constant time whether one is an ancestor of another in this tree, or if not whether one is clockwise of the other, by examining the cyclic ordering of vertices in their funnels. Construct the trees of junctions and their planar embedding in this way. The result of this stage of the algorithm should be a planar embedding of part of the canonical diagram consisting of all vertices and junctions, and the subset of the arcs that are part of a path from a junction to one of its funnel vertices. Check that the embedding is planar by computing its Euler characteristic, and abort the algorithm if it is not.
\item For each face $f$ of the drawing created in step~\ref{algorithm funnel junctions}, and each pair $j,j'$ of junctions belonging to $f$, use the data structure from step~\ref{algorithm interval adjacency} to test whether there is an edge whose trail passes through both $j$ and $j'$. This results in a graph $H_f$ in which the vertices represent the vertices or junctions on the boundary of $f$ and the edges represent pairs of vertices or junctions that must be connected, either by an arc or by shared membership in a marked face. The remaining arcs to be drawn in $f$ will be exactly the edges of $H_f$ that are not crossed by other edges of $H_f$; the marked faces in $f$ will be exactly the faces %
that contain pairs of crossing edges of $H_f$.
\item Within each face $f$ of the drawing so far, build a table using the same construction as in step~\ref{algorithm interval adjacency} that can be used to determine the existence of a crossing edge for an edge in $H_f$ in constant time. Use this data structure to identify the crossed edges, and draw an arc in $f$ for each uncrossed edge. For each face $g$ of the resulting subdivision of $f$, if $g$ has four or more vertices or junctions, find two pairs that would cross and test whether both pairs correspond to edges in $H_f$; if so, mark~$g$.
\item Construct a directed graph that has a vertex for each vertex of $G$, two vertices for each junction of the diagram (one in each direction), two directed edges for each arc, and a directed edge for each ordered pair of sharp angles that are non-consecutive in a marked face. By performing a depth-first search in this graph, determine whether there exist multiple smooth paths in the resulting drawing from any vertex of $G$ to any other point in the drawing, and abort the algorithm if any such pair of paths is found. Determine the set of vertices of $G$ reachable from $v$ and verify that it is the same set of vertices that are reachable in the original graph. Additionally, verify that the diagram satisfies the requirements in the definition of a canonical diagram. Abort the algorithm if any inconsistency is found in this step.
\item Convert the canonical diagram into a confluent drawing and return it.
\end{enumerate}

\begin{theorem}
For a given $n$-vertex graph $G$, and a given circular ordering of its vertices, it is possible to determine whether $G$ has an outerplanar strict confluent drawing with the given vertex ordering, and if so to construct one, in time $O(n^2)$.
\end{theorem}

\subsection{Drawings with low curve complexity}\label{sec:curvecomplexity}

Suppose that we are given a topological description of an outerplanar strict confluent drawing $D$ of a connected graph $G$, describing the tangency pattern and ordering of the arcs at each junction. It still remains to draw $D$ (or possibly an equivalent but combinatorially different outerplanar strict confluent drawing) in the plane using concrete curves for its arcs. If we ignore the tangency requirements at its junctions, the arcs and junctions of $D$ form a planar graph, but applying standard planar graph drawing methods will generate arcs that may not be smooth and that are not tangent to each other at the junctions. So how are we to draw $D$? In this section we use a circle packing method to draw $D$ with a small number of circular arcs for each arc of~$D$.
Thus, these drawings have low \emph{curve complexity} in the sense of Bekos et al.~\cite{BekKauKob-GD-12}, but with this complexity measured along arcs of the confluent diagram rather than edges of another type of graph drawing.

Given such a drawing $D$, let $D'$ be a modified version of $D$ in which every junction is incident to exactly three arcs, formed from $D$ by suppressing two-arc junctions and splitting junctions with more than three arcs.
Assume also (again by adding more junctions if necessary) that each vertex in $D'$ has only a single arc incident to it.

Given the topological diagram $D'$, we form a planar graph $H$ that has a vertex for each vertex or junction of $D'$, and an edge for each arc of $D'$.  Additionally, we create an edge in $H$ for each two vertices that are consecutive in the cyclic ordering of the vertices around the disk containing the drawing.

\begin{lemma}
\label{lem:3-regular}
$H$ is planar, 3-regular, and 3-vertex-connected.
\end{lemma}

\begin{proof}
Planarity and 3-regularity follow immediately from the construction of $H$. Every two vertices of $G$ are connected by three vertex-disjoint paths in $H$: at least one (not necessarily a smooth path) through $D$, using the assumption that $G$ is connected, and two more around the boundary of the disk. Therefore, if $H$ were not 3-vertex-connected, only one of its 3-connected components could contain vertices of $G$. The other components would either contain components of $D$ that are not part of any smooth path between vertices of $G$ (forbidden in a strict confluent drawing) or would contain more than one smooth path between the same sets of vertices (also forbidden).
\end{proof}

\begin{theorem}
Let $D$ be an outerplanar strict confluent drawing of a graph $G$, given topologically but not geometrically. Then we can construct an outerplanar strict confluent drawing of $G$ in which each arc of the drawing is represented by a smooth curve that is either a circular arc or the union of two circular arcs.
\end{theorem}

\noindent\textit{Proof.}
By the Koebe--Thurston--Andreev circle packing theorem, there exists a system $C$ of circles representing the faces of $H$, such that two circles are adjacent exactly when the corresponding faces share an edge. We may assume (by performing a M\"obius transformation if necessary) that the outer circle of this circle packing corresponds to the outer face of $H$. $C$ may be found efficiently (although not in strongly polynomial time) by a numerical iteration that quickly converges to the system of radii of the circles, from which their centers can also be computed easily~\cite{ColSte-CGTA-03,Moh-DM-93}.

\begin{wrapfigure}[20]{r}{0.5\textwidth}
    \vspace{-\baselineskip}
    \centering
    \includegraphics[scale=0.4] {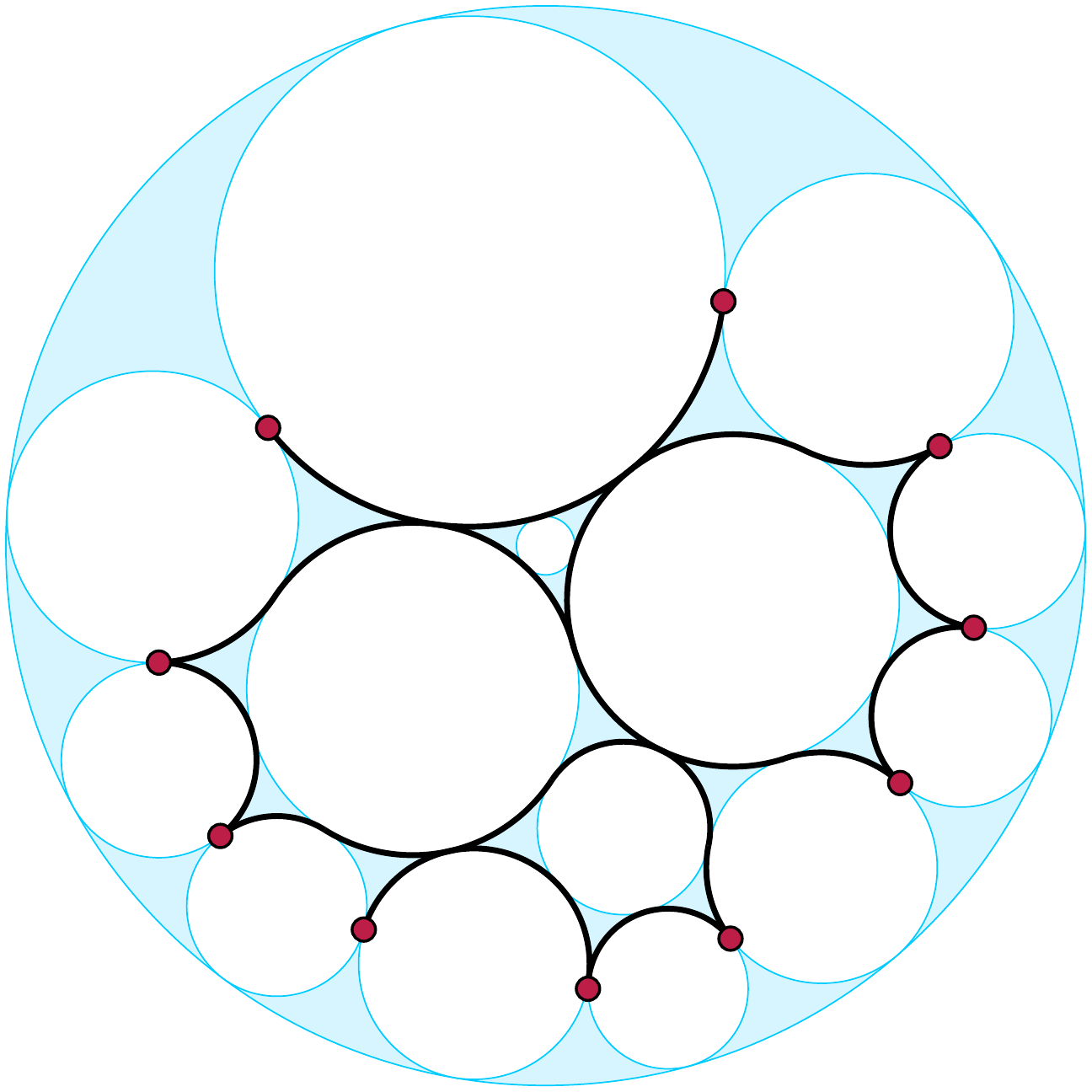}
	\caption{Construction of an outerplanar strict confluent drawing from a circle packing. The vertices of the drawing correspond to triangular gaps adjacent to the outer circle, and the junctions  correspond to the remaining triangular gaps.}
	\label{fig:confluent-circles}
\end{wrapfigure}
Each vertex of $G$ corresponds in $C$ to one of the triangular gaps between the outer circle and two other circles, and may be placed at the point of tangency of the two non-outer circles (one of the vertices of this triangle); see Fig.~\ref{fig:confluent-circles}.
The junctions in $D'$ lie at the meeting point of three faces of $H$, and correspond in $C$ to the remaining triangular gaps between three circles. A confluent drawing of $G$ may be formed by removing the outer circle, removing all circular arcs bounding the triangular gaps incident to the outer circle, and in each remaining triangular gap removing the arc that is on the other side of the sharp angle. The resulting drawing contracts some edges of $D'$ to form junctions with four incident arcs, but this does not affect the correctness of the drawing.
In the resulting drawing, arcs of the diagram that have merge points or vertices at both of their endpoints are drawn as two circular arcs (possibly both from the same circle); other arcs of the diagram are drawn as a single circular arc.\qed

\section {Conclusions}

We have shown that, in confluent drawing, restricting attention to the strict drawings allows us to completely characterize their complexity, and we have also shown that outerplanar strict confluent drawings with a fixed vertex ordering may be constructed in polynomial time. The most pressing problem left open by this research is to recognize the graphs that have outerplanar strict confluent drawings, without imposing a fixed vertex order. Can we recognize these graphs in polynomial time?

\medskip
\paragraph{Acknowledgements}
This work originated at Dagstuhl seminar 13151, \emph {Drawing Graphs and Maps with Curves}. D.E. was supported in part by the National Science Foundation under grants 0830403 and 1217322, and by the Office of Naval Research under MURI grant N00014-08-1-1015. M.L. was supported by the Netherlands Organisation for Scientific Research (NWO) under grant 639.021.123. M.N. received financial support by the `Concept for the Future' of KIT under grant YIG 10-209.

\small
\bibliographystyle{abbrv}
\bibliography{confluence}

\begin{thebibliography}{10}

\bibitem{BekKauKob-GD-12}
M.~A. Bekos, M.~Kaufmann, S.~G. Kobourov, and A.~Symvonis.
\newblock Smooth orthogonal layouts.
\newblock In {\em Graph Drawing 2012}, volume 7704 of {\em LNCS}, pages
  150{--}161. Springer, 2013.

\bibitem{ColSte-CGTA-03}
C.~R. Collins and K.~Stephenson.
\newblock A circle packing algorithm.
\newblock {\em Comput. Geom. Theory Appl.}, 25(3):233{--}256, 2003.

\bibitem{Cui2008}
W.~Cui, H.~Zhou, H.~Qu, P.~C. Wong, and X.~Li.
\newblock {Geometry-based edge clustering for graph visualization.}
\newblock {\em IEEE TVCG}, 14(6):1277--84, 2008.

\bibitem{degm-cd-05}
M.~Dickerson, D.~Eppstein, M.~T. Goodrich, and J.~Y. Meng.
\newblock Confluent drawings: Visualizing non-planar diagrams in a planar way.
\newblock {\em J. Graph Algorithms Appl.}, 9(1):31--52, 2005.

\bibitem{Dwyer2007}
T.~Dwyer, K.~Marriott, and M.~Wybrow.
\newblock Integrating edge routing into force-directed layout.
\newblock In {\em Graph Drawing 2006}, volume 4372 of {\em LNCS}, pages 8--19.
  Springer, 2007.

\bibitem{egm-dcd-06}
D.~Eppstein, M.~T. Goodrich, and J.~Y. Meng.
\newblock Delta-confluent drawings.
\newblock In {\em Graph Drawing 2005}, volume 3843 of {\em LNCS}, pages
  165--176. Springer, 2006.

\bibitem{EppGooMen-Alg-07}
D.~Eppstein, M.~T. Goodrich, and J.~Y. Meng.
\newblock {Confluent layered drawings}.
\newblock {\em Algorithmica}, 47(4):439{--}452, 2007.

\bibitem{EppSim-GD-11}
D.~Eppstein and J.~A. Simons.
\newblock Confluent {H}asse diagrams.
\newblock In {\em Graph Drawing 2011}, volume 7034 of {\em LNCS}, pages
  2{--}13. Springer, 2012.

\bibitem{hmr-becgcd-06}
M.~Hirsch, H.~Meijer, and D.~Rappaport.
\newblock Biclique edge cover graphs and confluent drawings.
\newblock In {\em Graph Drawing 2006}, volume 4372 of {\em LNCS}, pages
  405--416. Springer, 2007.

\bibitem{Holten2006}
D.~Holten.
\newblock {Hierarchical edge bundles: visualization of adjacency relations in
  hierarchical data.}
\newblock {\em IEEE TVCG}, 12(5):741--8, 2006.

\bibitem{Holten2009}
D.~Holten and J.~J. van Wijk.
\newblock {Force-Directed Edge Bundling for Graph Visualization}.
\newblock {\em Computer Graphics Forum}, 28(3):983--990, 2009.

\bibitem{hss-ttcd-04}
P.~Hui, M.~J. Pelsmajer, M.~Schaefer, and D.~\v{S}tefankovi\v{c}.
\newblock Train tracks and confluent drawings.
\newblock {\em Algorithmica}, 47(4):465--479, 2007.

\bibitem{Hurter2012}
C.~Hurter, O.~Ersoy, and A.~Telea.
\newblock {Graph Bundling by Kernel Density Estimation}.
\newblock {\em Computer Graphics Forum}, 31(3pt1):865--874, 2012.

\bibitem{l-pftu-82}
D.~Lichtenstein.
\newblock Planar formulae and their uses.
\newblock {\em SIAM J. Comput.}, 11(2):329--343, 1982.

\bibitem{Moh-DM-93}
B.~Mohar.
\newblock A polynomial time circle packing algorithm.
\newblock {\em Discrete Math.}, 117(1{--}3):257{--}263, 1993.

\bibitem{qa-cdard-10}
G.~Quercini and M.~Ancona.
\newblock Confluent drawing algorithms using rectangular dualization.
\newblock In {\em Graph Drawing 2010}, volume 6502 of {\em LNCS}, pages
  341--352. Springer, 2011.

\end{thebibliography}
\normalsize
\appendix\clearpage

\section{Combinatorial Complexity} \label {app:complexity}

We first assume that every vertex has degree one and every junction has degree three in the confluent drawing. Note that it is always possible to expand a confluent drawing to satisfy these restrictions. This expansion maintains the number of faces, and can only increase the number of arcs and junctions. Afterwards we will also consider junctions and vertices of higher degrees.
We also use the following lemma to reduce the number of cases we need to consider.

\begin {lemma} \label {lem:noloops}
  Let $D$ be a strict confluent drawing. $D$ cannot contain any smooth closed curves.
\end {lemma}

\begin {proof}
  Such a curve would either be irrelevant or would allow multiple paths (looping around the curve multiple times) for some pair of vertices.
\end {proof}

The following simple lemma is crucial for our bounds.
\begin{lemma}\label{lem:threesharp}
Every face in a strict confluent drawing must have at least three sharp angles.
\end{lemma}
\begin{proof}
Consider a face with less than three sharp angles. If a sharp angle is at a junction, then the arc opposite of the sharp angle must eventually reach a vertex, since by Lemma~\ref {lem:noloops} the drawing cannot have loops. Hence we can assume that there is a vertex at each sharp angle. But then the face must have at least three vertices, for otherwise it contains a self-loop (one vertex) or multiple paths between the same pair of vertices (two vertices).
\end{proof}
In the following, $n$ is the number of vertices, $k$ is the number of junctions, $m$ is the number of arcs, $F$ is the number of faces, and $c$ is the number of sharp angles. We first bound the combinatorial complexity of strict confluent drawings.
\begin{lemma}\label{lem:boundplanar}
Every strict confluent drawing has at most $2n - 4$ faces, $5n - 12$ junctions, and $8n - 18$ arcs.
\end{lemma}
\begin{proof}
By double counting we get that $2m = n + 3k$. The total number of sharp angles is $c = n + k$, and Lemma~\ref{lem:threesharp} implies that $c \geq 3F$, so that $n + k \geq 3F$. By combining the above relations with Euler's formula $n + k - m + F = 2$, we directly obtain that $F \leq 2n - 4$, $k \leq 5n - 12$, and $m \leq 8n - 18$.
\end{proof}
We can obtain similar bounds for outerplanar strict confluent drawings. Here we use $F_{in}$ to denote the number of internal faces ($F_{in} = F - 1$).
\begin{lemma}\label{lem:boundouterplanar}
Every outerplanar strict confluent drawing has at most $n - 2$ internal faces, $3n - 6$ junctions, and $5n - 9$ arcs.
\end{lemma}
\begin{proof}
By double counting we get that $2m = n + 3k$. The total number of sharp angles in internal faces is $c = k$ (every vertex is on the outer face), and Lemma~\ref{lem:threesharp} implies that $c \geq 3F_{in}$, so that $k \geq 3F_{in}$. By combining the above relations with Euler's formula $n + k - m + F_{in} = 1$, we directly obtain that $F_{in} \leq n - 2$, $k \leq 3n - 6$, and $m \leq 5n - 9$.
\end{proof}

\smallskip\noindent
{\bf Complex junctions.} If we allow junctions and vertices to have higher degree, then we can consider \emph{minimal drawings}, that is, confluent drawings with as few arcs and junctions as possible. For strict confluent drawings that are minimal we can obtain stronger upper bounds on the combinatorial complexity. Consider a junction of degree three. We call the arc opposite of the sharp angle the \emph{free arc}. Note that, regardless of what is on the other side of a free arc, we can always contract a free arc without changing the underlying graph represented by the confluent drawing. Furthermore, contracting a free arc will not influence other free arcs. That means that we can contract all free arcs simultaneously. Now consider the unique arc incident to a vertex. If this arc is not free, then the sharp angle of the adjacent junction must be bounded by this arc. In this case we call the respective sharp angle \emph{pointless}. This angle lies in the face that contains the vertex, but does not act as a ``corner'' of that face, and hence does not count towards the at least three sharp angles necessary for a face in a strict confluent drawing, as required by Lemma~\ref{lem:threesharp}. Let $k'$ be the total number of pointless sharp angles. Following the arguments of the proof of Lemma~\ref{lem:boundplanar} and instead using $n + k - k' \geq 3F$, we obtain bounds $k \le 5n - 12 - 2k'$ and $m \le 8n - 18 - 3k'$ for junctions and arcs, respectively. Note that at least $n - 2k'$ vertices must be adjacent to a free arc. Since every arc can be shared by at most two junctions/vertices, there must be at least $(k + n - 2k')/2$ free arcs. Hence there can be at most $k - (k + n - 2k')/2 \le 2n - 6$ junctions and at most $m - (k + n - 2k')/2 \le 5n - 12$ arcs (using $2m = n + 3k$) in a minimal strict confluent drawing.

For outerplanar strict confluent drawings, note that all pointless sharp angles lie in the outer face. Using the same reasoning as above with the bounds of Lemma~\ref{lem:boundouterplanar} (using $k - k' \geq 3F_{in}$), we obtain the following lemma.
\begin{lemma}\label{lem:boundminimal}
Every strict confluent drawing that is minimal has at most $2n - 6$ junctions and $5n - 12$ arcs. Every outerplanar strict confluent drawing that is minimal has at most $n - 3$ junctions and $3n - 6$ arcs.
\end{lemma}
Although we cannot prove a tight bound for strict confluent drawings, the above bound for outerplanar strict confluent drawings is in fact tight.
\begin{lemma}
Every outerplanar strict confluent drawing of a clique on $n \geq 3$ vertices has $n-2$ internal faces, and at least $n - 3$ junctions and $3n - 6$ arcs.
\end{lemma}
\begin{proof}
It is easy to see that every outerplanar strict confluent drawing of a graph with a triangle must contain an internal face $f$ with three junctions. Also, since the underlying graph is a clique, the drawing cannot contain outside connections to an arc of $f$, as the origin of this connection will not be able to reach the junction of $f$ opposite of the arc. Hence, every vertex must reach $f$ through one of the three junctions. This means we can subdivide the drawing into three parts, each a clique including one of the junctions of $f$ of size at least two. This leads to the following general recurrence relation
\begin{align*}
T(2) &= x\\
T(n) &= y + T(n_1 + 1) + T(n_2 + 1) + T(n_3 + 1)\\
&\text{ where } n = n_1 + n_2 + n_3 \text{ and } n_1,n_2,n_3 \geq 1
\end{align*}
This recurrence solves to $T(n) = x (2n - 3) + y (n-2)$. Note that every arc of the face $f$ must be present in the drawing and cannot be contracted. For the number of internal faces we use $(x,y)=(0,1)$ to obtain $n-2$ internal faces. For the number of arcs we use $(x,y)=(0,3)$ to obtain $3n - 6$ arcs. Finally, for the number of junctions/vertices we use $(x,y)=(1,0)$ to obtain $2n - 3$ junctions/vertices, which implies $n - 3$ junctions. For the last two bounds it is important that every subproblem on two elements involves at least one junction, which is true if we start with $n \geq 3$ vertices.
\end{proof}

\section{Properties of Canonical Diagrams} \label {app:canonical}

\begin{lemma}
\label{lem:consecutive}
In every outerplanar strict confluent drawing or canonical diagram, in which each vertex has at least one incident edge, there must be a pair of  vertices that are consecutive on the outer face of the drawing and adjacent in the corresponding graph. If there are at least three vertices, then there must be at least two such pairs.
\end{lemma}

\begin{proof}
Let $uv$ be an adjacent pair of vertices that are as close as possible to each other on the outer face, as measured by the smaller of the two sequences of vertices from $u$ to $v$ and from $v$ to $u$ around the outer face.
Then $uv$ must be consecutive. For, if they were not consecutive, then there would be a vertex $w$ between them.  Any trail from $w$ to one of its neighbors would have to cross the trail for $uv$, causing $w$ to be adjacent to one of $u$ or $v$. But this would contradict the choice of $uv$ as being as close as possible for an adjacent pair.

Next, suppose that there are three or more vertices and let $u'v'$ be an adjacent pair of vertices that are as close as possible to each other by a different distance, the size of the sequence of vertices from $u'$ to $v'$ or $v'$ to $u'$ around the outer face, whichever of these two sequences does not contain the consecutive pair $uv$. Note that $u'v'$ cannot equal $uv$, because $uv$ is the pair with the largest distance by this measure. Again, we claim that $u'v'$ must be consecutive, for if they were separated by another vertex $w'$ then the trail from $w'$ to one of its neighbors would have to cross the trail for $u'v'$, causing $w'$ to be adjacent to one of $u'$ or $v'$ and leading to a contradiction.
\end{proof}

\begin{lemma}
\label{lem:marked-tree}
In every canonical diagram the dual graph of the set of marked faces forms a forest.
\end{lemma}

\begin{proof}
Suppose for a contradiction that this dual graph contains an induced cycle $C$. It is not possible for marked faces to entirely surround a single junction of the diagram, because they all have sharp angles at the junction and every junction has two non-sharp angles. Therefore, the part of the diagram inside $C$ must consist of one or more faces bounded by arcs of the faces in $C$. These surrounded faces may form a single connected region $R$, or they may form multiple connected regions separated by junctions that appear more than once along the boundary of $C$; in the latter case, these regions and the junctions that connect them form a tree, and we may choose $R$ to be a leaf of the tree. Thus, in either of these two cases there exists a region $R$ consisting of a set of neighboring faces of the diagram, in which all but at most one of the junctions on the boundary of $R$ span an angle of $2\pi$ (the one exceptional junction being the one that connects $R$ to other connected regions within $C$).

The part of the diagram within $C$ can itself be viewed as a confluent drawing, within which each of the junctions that spans an angle of $2\pi$ is connected to at least one other junction (otherwise the arcs into that junction could not be part of a complete trail). By Lemma~\ref{lem:consecutive}, there exist two consecutive junctions on the boundary of $C$ that are connected by a smooth path within $R$. This path, together with a smooth path connecting the same two junctions within the marked face of $C$ that contains them both, forms a continuous smooth loop in the original diagram, contradicting Lemma~\ref {lem:noloops}. This contradiction shows that $C$ cannot exist.
\end{proof}

\eenplaatje[width=\textwidth] {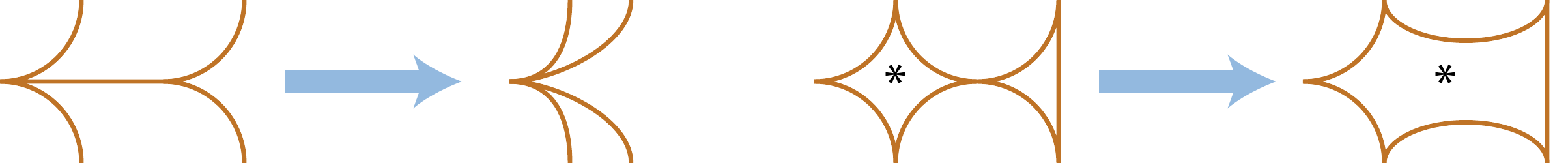} {The simplification operations that transform an outerplanar strict confluent drawing into its canonical diagram: contraction of arcs without sharp angles at both ends (left) and merger of marked faces and triangles (right).}

\begin{lemma}
\label{lem:canonize}
A graph $G$ may be represented by a canonical diagram if and only if it may be represented by an outerplanar strict confluent drawing.
\end{lemma}

\begin{proof}
To form a canonical diagram from a confluent drawing of $G$, repeatedly perform the following two simplifications (Figure~\ref{fig:canonize}):
\begin{itemize}
\item Contract any arc that does not have sharp angles or vertices at both of its ends.
\item Let $j$ be any junction with exactly two arcs in each direction, let $f$ and $f'$ be the two faces with sharp angles at $j$, and suppose that $f$ and $f'$ are both either marked or a triangle. Then merge the two faces $f$ and $f'$ by connecting them through $j$, removing the junction from the drawing, and mark the resulting merged face.
\end{itemize}
These operations preserve the properties that each arc is part of a trail, that each edge of $G$ is represented by a combinatorially unique trail, that each non-adjacent pair in $G$ has no trail, and that each marked face have the correct shape. Each step reduces the number of arcs, so this simplification process eventually terminates with a canonical diagram.

Conversely, any canonical diagram can be converted to an equivalent confluent drawing by repeatedly reversing the second of the two simplification operations. The reverse of this operation splits a marked face $f$ by pinching together two nonadjacent sides to form a new junction with four incident arcs; the two faces formed from $f$ are marked if they have more than three angles, or left unmarked if they are triangles. It is not possible to perform this pinching off step if one or both of the sides of the step is the boundary of another marked face. However, by Lemma~\ref{lem:marked-tree} it is always possible to choose a marked face to simplify that forms an isolated vertex or leaf of the dual graph of the set of remaining marked faces; such a face can always safely be simplified.
\end{proof}

In a canonical diagram $D$, define a \emph{pseudotriangle} for a set $K$ of vertices to be a face of $D$ that has exactly three sharp junctions, each of which is part of a trail from the face to a vertex in $K$, and for which all other boundary junctions (if they exist) do not lead to vertices in $K$. A \emph{side} of a pseudotriangle is the path of arcs connecting two of its sharp angles.

\begin{lemma}
\label{lem:candy-clique}
Let $G$ be represented by a canonical diagram $D$ and let $K$ be a clique of two or more vertices in $G$.
Then the arcs and faces traversed by trails connecting vertices of $K$ form a tree of arcs, marked faces, and pseudotriangles for $K$, connected to each other at junctions. In this tree, each junction is incident to exactly two arcs, marked faces, or pseudotriangles, one in each direction.
\end{lemma}

\begin{proof}
We use induction on the size of $K$; as a base case, the result is clearly true for $|K|=2$, for which the single trail forms a path of arcs and marked faces. If $|K|>2$ and $v$ is an arbitrary vertex of $K$, then the trails through $K-v$ form a tree of arcs, marked faces, and pseudotriangles by the induction hypothesis.
In order to reach every vertex of $K-v$, and avoid making multiple connections to any vertex, the trails from $v$ to $K-v$ can only connect to this tree in one of three ways:
\begin{enumerate}
\item It may be incident to one of the sharp angles of a marked face of $D$ that already belongs to  the tree.
\item It may be incident to a sharp angle of a marked face of $D$ that is not already in the tree, but that has an arc in the tree, causing the face to be added to the tree.
\item It may be incident to a sharp angle of a pseudotriangle of $D$ that is not already in the tree, but whose opposite side is already in the tree, again causing the face to be added to the tree.
\end{enumerate}
In any of these cases, outside of the tree for $K-v$, all the trails to $v$ must follow the same sequence of arcs and marked faces, for to do otherwise would violate the uniqueness of trails in canonical diagrams. The tree for $K-v$, together with the arcs and marked faces traversed by the rest of the trail from this tree to $v$, forms another structure of the same type, and contains trails from all vertices of $K-v$ to $v$.
\end{proof}

\begin{lemma}
\label{lem:funnel-nonadj}
For every junction $j$ of a canonical diagram with funnel intervals $[a,b]$ and $[c,d]$, at least one of the intervals must be separated.
\end{lemma}
\begin{proof}
For the sake of contradiction, assume that $[a,b]$ and $[c,d]$ are not separated. Then the vertices $a$, $b$, $c$, and $d$ form a clique. By Lemma~\ref{lem:candy-clique} the drawing
would have two pseudotriangles or marked faces, one on each side of $j$.
The boundaries of these faces adjacent to $j$ can be the only arcs into $j$, for otherwise one of $a$, $b$, $c$, or $d$ would not be in the funnel (see Fig.~\ref{fig:FunnelSep}). Now assume there is a junction $j'$ between $a$ and $b$ (or $c$ and $d$) on the boundary of the face opposite to $j$ (in this case the face must be a pseudotriangle). Consider the path that leaves $j'$ in the direction of $a$ and continues as far clockwise as possible. This path cannot pass through the pseudotriangle, because that would result in multiple paths between $b$ and either $c$ or $a$, which is not allowed. Hence this path must reach $a$, for otherwise $a$ would not be a funnel vertex of $j$. Similarly, the path that leaves $j'$ in the direction of $b$ and continues as counterclockwise as possible must reach $b$. The other paths of the funnel must end at a vertex in the circular interval $[a,b]$. This would imply that $[a,b]$ is separated, which contradicts our assumption. Thus, the faces on either side of $j$ must be marked faces or triangles. But this configuration violates the condition that there be no junction $j$ with four arcs in which the two faces having sharp angles at $j$ are marked or triangles.
\end{proof}

\eenplaatje{FunnelSep} {The configuration of Lemma~\ref{lem:funnel-nonadj}.}

\begin{lemma}
\label{lem:candy-same-junctions}
Let $D$ and $D'$ be canonical diagrams for the same graph $G$ with the same vertex ordering.
Then for every junction $j$ of $D$, there must be a junction $j'$ in $D'$ with the same funnel.
\end{lemma}

\begin{proof}
For a junction $j$ of $D$ with funnel intervals $[a,b]$ and $[c,d]$, let $[a,b]$ be the separated funnel interval. We prove the lemma by induction on the cardinality of $[a,b]$. The existence of the junction implies that $(a, c)$ and $(b, d)$ are edges of $G$. Therefore, the trail in $D'$ between $a$ and $c$ must cross the trail between $b$ and $d$ in the canonical diagram. This can happen in three ways (see Fig.~\ref{fig:JunctionCases}): (1) the trails from $a$ and $b$ merge and then split towards $c$ and $d$, (2) the trails from $b$ and $c$ merge and then split towards $a$ and $d$, or (3) the trails cross inside a marked face $F$ from which each vertex $a,b,c,$ and $d$ is reached by a different junction of $F$.

In case (1) the trails can travel together along a sequence $p$. But because $[a,b]$ and $[c,d]$ are the funnel intervals for $j$, there can be nothing entering or leaving $p$ between the merge and split points, and $p$ cannot contain a nonzero number of arcs without violating the condition on arcs without sharp angles at both ends. Therefore, the merge point of the trails from $a$ and $b$ is also the split point of the trails to $c$ and $d$, and forms a single junction $j'$ in $D'$ with the same funnel.

In the latter two cases, $a$ and $b$ are connected by a trail in $D'$, therefore they are adjacent in $G$.
Hence there must exist a junction $\iota$ in $D$ with funnel intervals $[a,e]$ and $[f,b]$, where $e, f \in [a,b]$. If $[a,b]$ contains less than two vertices besides $a$ and $b$, then this is not possible (base case). Otherwise, by induction, there must exist a corresponding junction $\iota'$ in $D'$ with funnel intervals $[a,e]$ and $[f,b]$.
This junction must lie along the trail between $a$ and $b$. By the properties of a funnel, neither $e$ nor $f$ can be connected to $c$ or $d$. This directly makes case (2) impossible, since such a connection is necessary in that configuration. In case (3) this implies that $\iota'$ must lie on the part of the trail between $a$ and $b$ that is along (or through) $F$. This is not possible if the trail passes through $F$, so $\iota'$ must be at the boundary of $F$. But this makes the marked face $F$ invalid, as it may have only sharp angles.
\end{proof}

\eenplaatje{JunctionCases} {The cases of Lemma~\ref{lem:candy-same-junctions}.}

\begin{lemma}
\label{lem:trail-thru-junction}
Let $j$ be a junction of canonical diagram $D$ with funnel $a,b,c,d$, and let $uv$ be an
edge of the graph $G$ represented by $D$. Then the trail for edge $uv$ passes through $j$ if and only if
exactly one of $u$ and $v$ is in the closed interval $[a,b]$ and exactly one
of $u$ and $v$ is in the closed interval $[c,d]$.
\end{lemma}

\begin{proof}
If the trail
passed through $j$ but $u$ or $v$ was outside these intervals, it would be
reached from $j$ by a more extreme path than one of $a$, $b$, $c$, or $d$,
violating the definition of a funnel. If both are inside the same interval, but the trail between them goes through $j$, then the diagram violates the requirement that trails be unique.
And if $u$ and $v$ are both inside
different intervals, then the path must either go through $j$ or it must
cross the funnel twice, again causing a violation of the requirement that trails be unique.
\end{proof}

\begin{lemma}
\label{lem:trail-across-mark}
Let $D$ be a canonical diagram in which two junctions $j$ and $j'$ are consecutive on the trail from $u$ to $v$. Then this trail passes across a marked face from $j$ to $j'$ if and only if there exist vertices $p$ and $q$, in the cyclic order $u$, $p$, $v$, $q$, such that $u$, $v$, $p$, and $q$ form a clique, and such that $p$ and $q$ are both in the same interval as $v$ in the funnel for $j$ and in the same interval as $u$ in the funnel for $j'$.
\end{lemma}

\begin{proof}
Suppose first that the trail crosses a marked face. This face must have sharp angles on both sides of trail $uv$; let $p$ and $q$ be any two vertices reachable from two angles on opposite sides.  Then there also exist trails from $u$ and $v$ to $p$ and $q$ that diverge from the trail from $u$ to $v$ within the marked face. The containment of $p$ and $q$ within the stated intervals follows from Lemma~\ref{lem:trail-thru-junction}.

In the other direction, suppose that the trail passes consecutively through $j$ and $j'$ and that there exist $p$ and $q$ satisfying the conditions of the lemma. Since $u$, $v$, $p$, and $q$ form a clique, Lemma~\ref{lem:candy-clique} implies that this clique is represented either by a marked face with at least four sharp angles connecting to all four of these vertices, or by two marked faces or triangles connected by a sequence of junctions, arcs, and additional marked faces. However, because of the assumption in the lemma that $p$ and $q$ belonging to certain intervals, all of these faces must lie between $j$ and $j'$; the additional assumption that $j$ and $j'$ are consecutive on the trail means that there is only one possibility, a single marked face with separate sharp angles connecting to the four vertices $u$, $v$, $p$, and $q$. Thus, as the lemma states, the trail crosses a marked face.
\end{proof}

\begin{lemma}
\label{lem:candy-same-arcs}
Let $D$ and $D'$ be canonical diagrams for the same graph $G$ with the same vertex ordering.
Then for every arc from junction $j$ to $j$' of $D$, there must be an arc in $D'$ connecting the corresponding junctions.
\end{lemma}

\begin{proof}
Let $uv$ be an edge of $G$ whose trail in $D$ uses the arc.
The trails in $D$ and $D'$ from $u$ to $v$ both go through the same sets of junctions by Lemma~\ref{lem:trail-thru-junction}.
On the trail in $D$ from $u$ to $v$, the junctions are monotonic in their ordering
by the size of the sets of reachable vertices in each direction, and the same is true in $D'$, so the trails go through the same junctions in the same ordering.
Therefore, in $D'$, the trail from $u$ to $v$ also goes from $j$ to $j'$ with no intervening junctions.
The only way it can avoid using an arc that satisfies the lemma is for it to cross a marked face instead of an arc, but this would violate Lemma~\ref{lem:trail-across-mark} for either $D$ or $D'$.
\end{proof}

\begin{theorem}
Every two canonical diagrams for the same graph $G$ with the same vertex ordering are isomorphic.
\end{theorem}

\begin{proof}
This follows from Lemma~\ref{lem:candy-same-junctions} (showing that they have the same set of junctions) and Lemma~\ref{lem:candy-same-arcs} (showing that they have the same set of arcs).
\end{proof}

\end{document}